\documentclass{article}
\usepackage{spconf,amsmath,graphicx}
\usepackage{amssymb}
\usepackage{amsmath}
\usepackage{amsthm}
\usepackage{cite}
\usepackage{bbm}
\usepackage{algpseudocode}
\usepackage{algorithm,algpseudocode}
\usepackage{algorithmicx}
\usepackage{url}
\usepackage{color}
\usepackage{enumitem}
\usepackage{cases}

\newtheorem{lemma}{Lemma}
\newtheorem{theorem}{Theorem}
\title{Controlling Smart Propagation Environments: Long-Term versus Short-Term Phase Shift Optimization \vspace*{-0.4cm}}
\name{\begin{tabular}{c}Trinh Van Chien$^\ast$, Lam Thanh Tu$^\xi$, Dinh-Hieu Tran$^\ast$, Hieu Van Nguyen$^\ast$,
		\\  Symeon~Chatzinotas$^\ast$, Marco~Di~Renzo$^{\dagger}$, and Bj\"{o}rn Ottersten$^\ast$ \end{tabular} \vspace*{-0.5cm} \thanks{This   work   was partially supported by RISOTTI-Reconfigurable  Intelligent  Surface  for  Smart  Cities  under  ProjectFNR/C20/IS/14773976/RISOTTI.}}
\address{\vspace*{-0.05cm}$^{\ast}$Interdisciplinary Centre for Security, Reliability and Trust, University of Luxembourg, Luxembourg\\
$^{\xi}$  Institute XLIM, University of Poitiers, France \\
$^{\dagger}$Universit\'e  Paris-Saclay,  CNRS,  CentraleSup\'elec, Laboratoire  des  Signaux  et  Syst\`emes,  France
\vspace*{-0.6cm}
}
\begin{document}
%
\maketitle
\begin{abstract}
 Reconfigurable intelligent surfaces (RISs) have recently gained significant interest as an emerging technology for future wireless networks. This paper studies an RIS-assisted  propagation environment, where a single-antenna source transmits data to a single-antenna destination in the presence of a weak direct link. We analyze and compare RIS designs based on long-term and short-term channel statistics in terms of coverage probability and ergodic rate. For the considered optimization designs, closed-form expressions for the coverage probability and ergodic rate are derived. We use numerical simulations to analyze  and compare against analytic results in finite samples. Also, we show that the considered optimal phase shift designs outperform several heuristic benchmarks.
\end{abstract}

\vspace{-0.15cm}
\begin{keywords}
Reconfigurable intelligent surface, coverage probability, ergodic channel rate
\end{keywords}
\vspace*{-0.5cm}
\section{Introduction}
\label{sec:intro}
\vspace*{-0.4cm}

Data throughput and coverage enhancements are of paramount importance in fifth-generation (5G) and beyond networks \cite{giordani2020toward}. In this context, reconfigurable intelligent surfaces (RISs) have received significant attention from academic and industrial researchers because of their ability to control the wireless propagation environment through passive reflecting elements integrated with low-cost electronics \cite{qian2020beamforming, wu2019intelligent, abrardo2020intelligent}. 

The complex nature of wireless environments results in propagation channels that are characterized by small-scale and large-scale fading. RISs aim at shaping the electromagnetic waves in complex wireless environments by appropriately optimizing the phases of their constitutive elements (i.e., the unit cells). Due to the small-scale and large-scale dynamics that characterize a complex wireless channel, the phase shifts of the RIS elements can be optimized based on different time scales \cite{Chien2021TWC, zhi2021two}. Most of the works in the literature have considered the optimization of the phase shifts by assuming the perfect knowledge of the instantaneous channel state information, and different performance metrics have been considered \cite{wu2019intelligent, zhang2020sum, zappone2020overhead}. This optimization criterion is based on adjusting the phase shifts of the RIS elements based on the small-scale dynamics of the channel, and, therefore, results in the best achievable performance. This optimization criterion may, however, not be applicable in some application scenarios that are characterized by a short coherent time,  since the optimal phase shifts of the RIS elements need be updated frequently in order to adapt to the rapid changes that characterize the small-scale fading \cite{jung2020performance}. 
Another option for optimizing the phase shifts of the RIS elements is based on leveraging only statistical channel state information (CSI), i.e., the large-scale characteristics of the wireless channel \cite{abrardo2020intelligent, 9140329}. 
Optimization criteria based on long-term CSI need to be updated less frequently, and this reduces the channel estimation overhead. However, some performance degradation is expected as compared with the optimal phase shift design based on instantaneous CSI.
Even though some research works have recently proposed the design of RISs based on long-term CSI, to the best of our knowledge, no previous work has comprehensively and analytically analyzed as well as compared the achievable performance of RIS-assisted wireless networks based on short-term and long-term CSI. 

Based on these considerations, the aim of this paper is to study the performance of RIS-assisted systems by considering optimization criteria based on long-term channel statistics and short-term channel statistics, i.e., based on prefect CSI.
We analyze the coverage probability and the ergodic channel capacity under long-term and short-term phase shift design criteria.
Specifically, we derive closed-form expressions of the coverage probability and the  ergodic channel rate for both optimization criteria, highlighting findings that have not been investigated before. Several insights are also observed, such as that the long-term phase shift design offers similar performance as the short-term optimal phase shift design as the number of RIS elements increases.

\textit{Notation}: Upper-bold and lower-bold letters are used to denote matrices  and vectors, respectively. The identity matrix of size $M \times M$ is denoted by $\mathbf{I}_M$. The Hermitian and regular transpose are denoted $(\cdot)^H$ and $(\cdot)^T$, respectively. $\mathcal{CN}(\cdot, \cdot)$ denotes a circularly symmetric Gaussian distribution. The expectation of a random variable is denoted by $\mathbb{E}\{ \cdot \}$. The upper incomplete Gamma function is denoted by $\Gamma(m,n) = \int_{n}^{\infty} t^{m-1} \mathrm{exp}(-t) dt$ and $\Gamma(x) = \int_{0}^{\infty} t^{x-1}  \mathrm{exp}(-t) dt$ denotes the Gamma function.
\vspace*{-0.5cm}
\section{System Model}
\label{sec:SysModel}
\vspace*{-0.5cm}
We consider an RIS-assisted communication system  where a single-antenna source communicates with a  single-antenna destination. A frequency-flat block-fading channel model is assumed in each coherence interval. 
The RIS comprises $M$ reflecting elements. The RIS phase shift matrix $\pmb{\Phi} \in \mathbb{C}^{M \times M}$ is defined as $\pmb{\Phi} = \mathrm{diag}\big([e^{j\theta_{1}}, \ldots, e^{j\theta_{M}}]^T \big)$, where $\theta_{m} \in [-\pi, \pi]$ is the phase shift of the $m$-th reflecting element.
\vspace*{-0.2cm}
\subsection{Channel Model}
\vspace*{-0.2cm}
The channel of the direct link between the source and the destination is $h_{\mathrm{sd}} \in \mathbb{C}$. 
The indirect link from the source to the destination comprises the channel between the source and the RIS, which is denoted by $\mathbf{h}_{\mathrm{sr}} \in \mathbb{C}^M$, and the channel between the RIS and the destination, which is denoted by $\mathbf{h}_{\mathrm{rd}} \in \mathbb{C}^M$. 
Specifically, the channels are defined as
	$h_{\mathrm{sd}} = \sqrt{\beta_{\mathrm{sd}}} g_{\mathrm{sd}}$, $\mathbf{h}_{\mathrm{sr}} = \bar{\mathbf{h}}_{\mathrm{sr}} + \mathbf{g}_{\mathrm{sr}}$, and $\mathbf{h}_{\mathrm{rd}} = \bar{\mathbf{h}}_{\mathrm{rd}} + \mathbf{g}_{\mathrm{rd}}$,
	where $g_{\mathrm{sd}} \sim \mathcal{CN}(0,1)$, $g_{\mathrm{sr}}  \sim \mathcal{CN}(0, \mathbf{I}_M \beta_{\mathrm{sr}} /(K_{\mathrm{sr}}+1))$, and $g_{\mathrm{rd}}  \sim \mathcal{CN}(0, \mathbf{I}_M \beta_{\mathrm{rd}} /(K_{\mathrm{rd}}+1))$ are the small-scale fading contributions; $\beta_{\mathrm{sd}},$ $\beta_{\mathrm{sr}},$ and $\beta_{\mathrm{rd}}$ are the large-scale fading coefficients; and $K_{\mathrm{sr}} \geq 0 $ and $K_{\mathrm{rd}} \geq 0$ are the Rician factors. Based on \cite{massivemimobook}, the line-of-sight (LoS) channel vectors $\bar{\mathbf{h}}_\alpha \in \mathbb{C}^{M} , \alpha \in \{\mathrm{sr}, \mathrm{rd} \},$ are given as follows
	\vspace*{-0.2cm}
	\begin{equation} \label{eq:barhsr}
		\bar{\mathbf{h}}_\alpha = \sqrt{\frac{K_\alpha \beta_\alpha }{K_{\alpha}+1}} \left[e^{j \mathbf{k}(\psi_{\alpha}, \phi_{\alpha} )^T \mathbf{u}_1}, \ldots,  e^{j \mathbf{k}(\psi_{\alpha}, \phi_{\alpha} )^T \mathbf{u}_M} \right]^T,
	\vspace*{-0.1cm}
	\end{equation}
where $\psi_{\alpha}$ and $\phi_{\alpha}$ are the azimuth and elevation angles of departure (AoD)  under which the RIS views the source and the destination for $\alpha = {\rm{sr}}$ and $\alpha = {\rm{rd}}$, respectively. By assuming that the RIS is a planar surface, the wave vectors in \eqref{eq:barhsr}, $\mathbf{k}(\psi_{\alpha}, \phi_\alpha)$, are 
\vspace*{-0.2cm}
\begin{equation} \label{eq:kvec}
	\mathbf{k}(\psi_{\alpha}, \phi_\alpha) = \frac{2\pi}{\lambda} \left[ z_1, \, z_2, \,\sin(\psi_{\alpha}) \right]^T,
\vspace*{-0.2cm}
\end{equation}
with $z_1 = \cos(\psi_{\alpha})\cos(\phi_\alpha)$ and $z_2 = \sin(\psi_{\alpha})\cos(\phi_\alpha)$, and $\lambda$ is the signal wavelength. 
Also, the vector $\mathbf{u}_m$ in \eqref{eq:barhsr} is defined as $\mathbf{u}_m =[0, \, \mathrm{mod}(m-1,M_H)d_r, \, \lfloor (m-1)/M_H \rfloor d_r ]^T$, where $\mathrm{mod}$ is the modulus operation and $\lfloor \cdot \rfloor$ is the floor function. $d_r$ is the element spacing at the RIS. 

To facilitate the analysis in the next section, Theorem~\ref{Theorem:ChannelStatistics} gives the moments of the RIS-assisted (cascaded) channel for an arbitrary phase shift matrix $\pmb{\Phi}$.
\vspace*{-0.2cm}
\begin{theorem} \label{Theorem:ChannelStatistics}
	The indirect link from the source to the destination through the RIS has the following statistical moments 
\vspace*{-0.4cm}
\begin{align} \label{eq:Expecs}
	\mathbb{E} \{ |\mathbf{h}_{\mathrm{sr}}^H \pmb{\Phi} \mathbf{h}_{\mathrm{rd}}|^2 \} & =   \delta , \;\;\; \mathbb{E} \{ |\mathbf{h}_{\mathrm{sr}}^H \pmb{\Phi} \mathbf{h}_{\mathrm{rd}}|^4 \}  =  \delta^2  + a,
\vspace*{-0.2cm}
\end{align}
where $\delta = |\bar{\alpha}|^2 + M\mu \widetilde{K}$, $\bar{\alpha} = \bar{\mathbf{h}}_{\mathrm{sr}}^H \pmb{\Phi} \bar{\mathbf{h}}_{\mathrm{rd}}$, $\mu = \beta_{\mathrm{sr}}\beta_{\mathrm{rd}}/\omega$, $\omega = (K_{\mathrm{sr}}+1)(K_{\mathrm{rd}}+1)$,  $\widetilde{K} = K_{\mathrm{sr}}+ K_{\mathrm{rd}}+1$, $\widehat{K} = 1+ 2K_{\mathrm{sr}} + 2K_{\mathrm{rd}}$, and $a =2M |\bar{\alpha}|^2 \mu  \widetilde{K} + M^2 \mu^2 \widetilde{K}^2 +  2M \mu^2 \widehat{K} +8 |\bar{\alpha}|^2 \mu$.
\vspace*{-0.2cm}
\end{theorem}
\begin{proof}
The proof follows by using known results on the moments of Rician random variables. It is omitted due to space limitations.
\end{proof}
The second and fourth moments demonstrate that the array gain due to the presence of an RIS is
proportional to $M$ and $M^2$, respectively.
\vspace*{-0.2cm}
\subsection{Phase Shift Designs and Channel Rate}
\vspace*{-0.2cm}
If the source transmits a data symbol $s$ with $\mathbb{E}\{ |s|^2 \} = 1$, the received signal $y \in \mathbb{C}$ at the destination is 
\vspace*{-0.2cm}
\begin{equation} \label{eq:ReceiveSig}
	y = \sqrt{\rho} \left( h_{\mathrm{sd}}  +  \mathbf{h}_{\mathrm{sr}}^H \pmb{\Phi} \mathbf{h}_{\mathrm{rd}} \right) s  + n,
\vspace*{-0.2cm}
\end{equation}
where $\rho$ is the transmit power and $n \sim \mathcal{CN}(0,\sigma^2)$ is the additive noise. The phase shift matrix $\pmb{\Phi}$ is usually optimized based on the CSI. In this paper, we focus our attention on two design criteria.
\vspace*{-0.2cm}
\begin{itemize}[leftmargin=*]
	\item[$i)$] \textit{Short-term phase shift design}: The phase shifts of the RIS elements are optimized based on perfect CSI, which encompasses large-scale and small-scale fading statistics.
	\vspace*{-0.2cm}
    \item[$ii)$] \textit{Long-term phase shift design}: The phase shifts of the RIS elements are optimized based on statistical CSI. In particular, the optimal phase shift matrix is obtained by maximizing the average SNR at the destination.
\vspace*{-0.2cm}
\end{itemize}
The short-term phase shift design corresponds to the best achievable performance.
Let us assume that we are interested in maximizing the received signal strength, i.e., $\rho  | h_{\mathrm{sd}}  +  \mathbf{h}_{\mathrm{sr}}^H \pmb{\Phi} \mathbf{h}_{\mathrm{rd}} |^2$. Then, the short-term optimal phase shift of the $m$-th RIS element, which is denoted by $\theta_{m}^{\mathsf{opt}, \mathsf{st}}$, is \cite{wu2019intelligent,van2021outage}
\vspace*{-0.2cm}
\begin{equation} \label{eq:Phase1}
	\theta_{m}^{\mathsf{opt}, \mathsf{st}} = \arg(h_{\mathrm{sd}}) - \arg([\mathbf{h}_{\mathrm{sr}}^\ast]_m)  - \arg([\mathbf{h}_{\mathrm{rd}}]_m), \forall m,
\vspace*{-0.1cm}
\end{equation}
where $[ \mathbf{h}_{\mathrm{sr}} ]_m$ and $[\mathbf{h}_{\mathrm{rd}}]_m$ are the $m$-th element of $\mathbf{h}_{\mathrm{sr}} $ and $\mathbf{h}_{\mathrm{rd}}$, respectively.   
The optimal phase shift of each RIS element in \eqref{eq:Phase1} needs to be updated every channel coherent interval. 
By contrast, the long-term phase shift matrix can be applied for a longer period of time, which spans many coherence intervals.
Conditioned on the phase shift matrix and the CSI, the channel rate is formulated as
\vspace*{-0.2cm}
\begin{equation} \label{eq:Rran} 
	R = \log_2 \left(1 +  \gamma \right), \mbox{[b/s/Hz]},
\vspace*{-0.1cm}
\end{equation}
where the signal-to-noise ratio (SNR) value $\gamma$ is
\vspace*{-0.2cm}
\begin{subnumcases}
{\gamma =}
\nu |h_{\mathrm{sd}} + \mathbf{h}_{\mathrm{sr}}^H \pmb{\Phi}^{\mathsf{opt}, \mathrm{lt}} \mathbf{h}_{\mathrm{rd}}|^2, & \mbox{Long-term}, \label{eq:SNRLT}\\
 \nu  \left( | h_{\mathrm{sd}} | + \sum\limits_{m = 1}^M  \omega_m \right)^2,&  \mbox{Short-term}, \label{eq:SNRST}
\end{subnumcases}
where $\omega_m = | [ \mathbf{h}_{\mathrm{sr}} ]_m | | [\mathbf{h}_{\mathrm{rd}}]_m |, \forall m$. The SNR value $\gamma$ that corresponds to the short-term phase shift design is obtained by using the optimal phase shift design in \eqref{eq:Phase1}. As far as the long-term phase shift design is concerned, the optimal phase shift matrix $\pmb{\Phi}^{\mathsf{opt}, \mathrm{lt}} = \mathrm{diag}\big([e^{j\theta_{1}^{\mathsf{opt}, \mathsf{lt}}}, \ldots, e^{j\theta_{M}^{\mathsf{opt}, \mathsf{lt}}}]^T \big)$ is obtained in 
Lemma~\ref{lemma:ChanStaDesign}.
\vspace*{-0.2cm}
\begin{lemma} \label{lemma:ChanStaDesign}
	If the $m$-th phase shift of the RIS is set as follows
\vspace*{-0.2cm}
	\begin{equation} \label{eq:PhaseLoS}
		\theta_{m}^{\mathsf{opt}, \mathsf{lt}} = -\arg([\bar{\mathbf{h}}_{\mathrm{sr}}^\ast]_m)  - \arg([\bar{\mathbf{h}}_{\mathrm{rd}}]_m) ,
\vspace*{-0.2cm}
	\end{equation}
	then the average received SNR, $\mathbb{E}\{\gamma \}$ with $\gamma$ given in \eqref{eq:SNRLT} is maximized.
	\vspace*{-0.2cm}
\end{lemma}
\begin{proof}
	The proof follows by maximizing the cost function $\mathbb{E}\{ \gamma\}$ subject to the phase shift constraints $\theta_m \in [-\pi, \pi], \forall m$. The detailed proof is omitted due to space limitations.
\end{proof}
The phase shift design in \eqref{eq:PhaseLoS} depends only on the LoS components of the channels.
The short-term and long-term phase shift designs are both aimed at boosting the strength of the received signal. The short-term phase shift design is, in general, an upper-bound for the long-term phase shift design. In analytical terms, in fact, we have the following  
\vspace*{-0.2cm}
\begin{equation} \label{eq:Ratebound}
	\begin{split}
		& \log_2(1 + \nu |h_{\mathrm{sd}} + \mathbf{h}_{\mathrm{sr}}^H \pmb{\Phi}^{\mathsf{opt}, \mathrm{lt}} \mathbf{h}_{\mathrm{rd}}|^2)  \stackrel{(a)}{\leq} \underset{\{ \theta_m \}}{\max}\, \log_2 ( 1 + \gamma) \\
		&\stackrel{(b)}{=} \log_2(1 + \nu |h_{\mathrm{sd}} + \mathbf{h}_{\mathrm{sr}}^H \pmb{\Phi}^{\mathsf{opt}, \mathsf{st}} \mathbf{h}_{\mathrm{rd}}|^2,
	\end{split}
\vspace*{-0.2cm}
\end{equation}
where  $\pmb{\Phi}^{\mathsf{opt}, \mathsf{st}} = \mathrm{diag}\big([e^{j\theta_{1}^{\mathsf{opt}, \mathsf{st}}}, \ldots, e^{j\theta_{M}^{\mathsf{opt}, \mathsf{st}}}]^T \big)$ is the short-term phase shift matrix  with $\theta_{m}^{\mathsf{opt}, \mathrm{st}}, \forall m,$ defined in \eqref{eq:Phase1}. 
In particular, $(a)$ is obtained because the phase shift solution that maximizes the average received SNR is a feasible point of the capacity maximization problem as a function of the instantaneous CSI, and $(b)$ is obtained because the short-term phase shift design in \eqref{eq:Phase1} is the optimal solution.
\vspace*{-0.3cm}
\section{Coverage Probability and Ergodic Channel Rate}
\vspace*{-0.3cm}
In this section, we introduce analytical frameworks for the coverage probability and the ergodic rate for the short-term and long-term phase shift designs.
\begin{figure*}[t]
	\begin{minipage}{0.33\textwidth}
		\centering
		\includegraphics[trim=0.8cm 0cm 1.3cm 0.6cm, clip=true, width=2.3in]{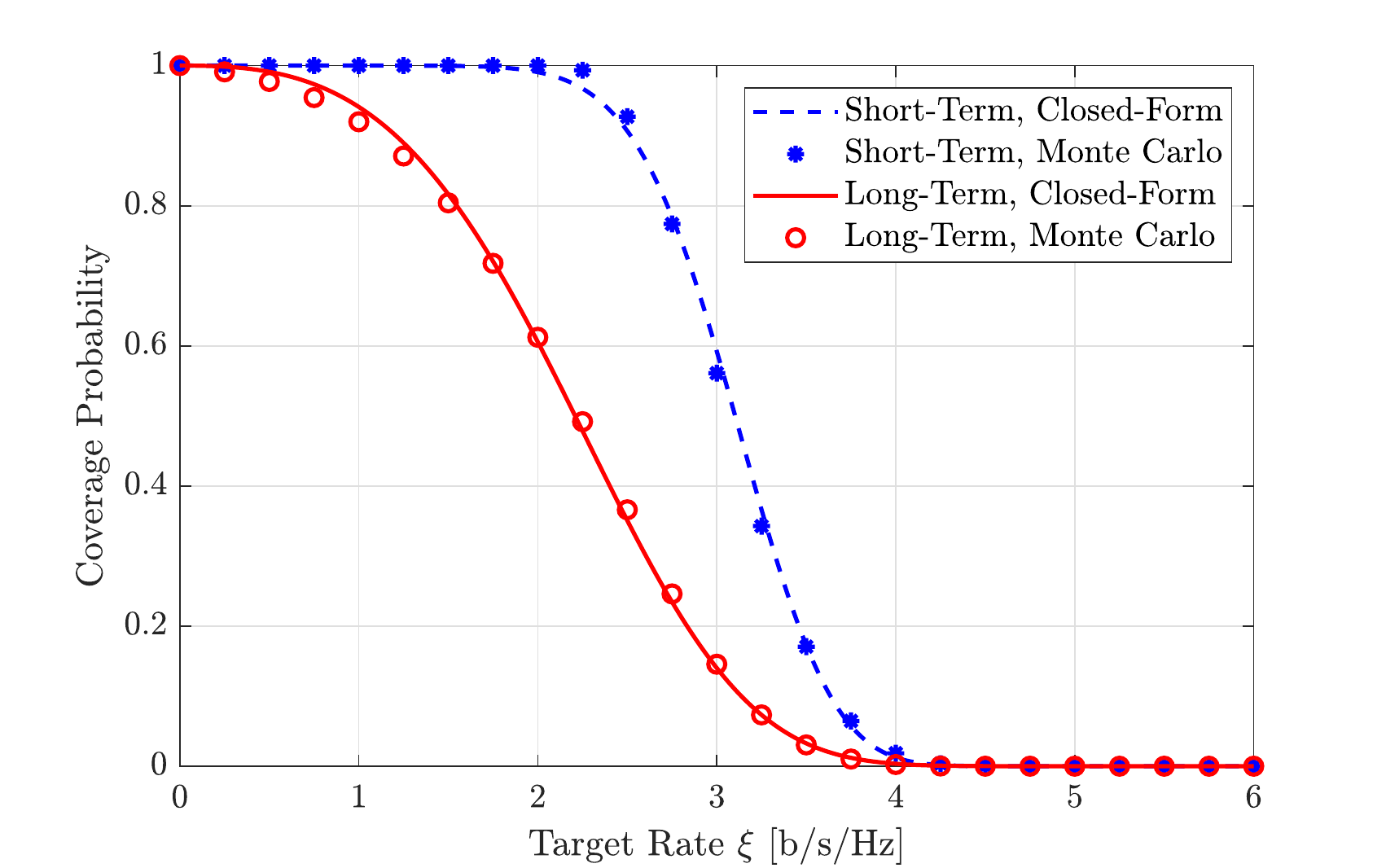} \vspace*{-0.4cm}\\
		(a)
		\vspace*{-0.3cm}
	\end{minipage}
	\begin{minipage}{0.33\textwidth}
		\centering
		\includegraphics[trim=0.8cm 0cm 1.3cm 0.5cm, clip=true, width=2.3in]{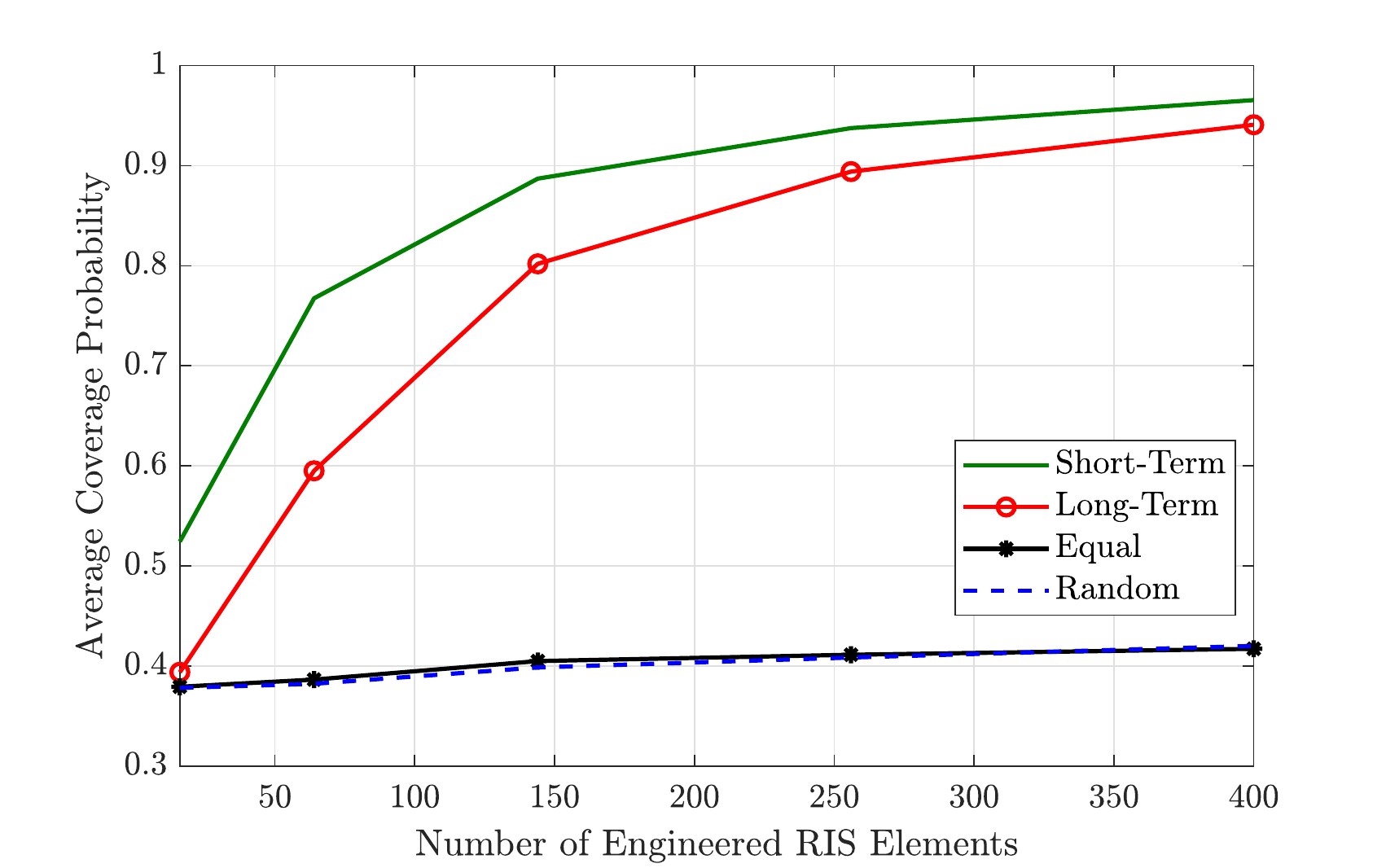} \vspace*{-0.4cm} \\
		(b)
		\vspace*{-0.3cm}
	\end{minipage}
	\begin{minipage}{0.33\textwidth}
		\centering
		\includegraphics[trim=0.8cm 0cm 1.3cm 0.5cm, clip=true, width=2.3in]{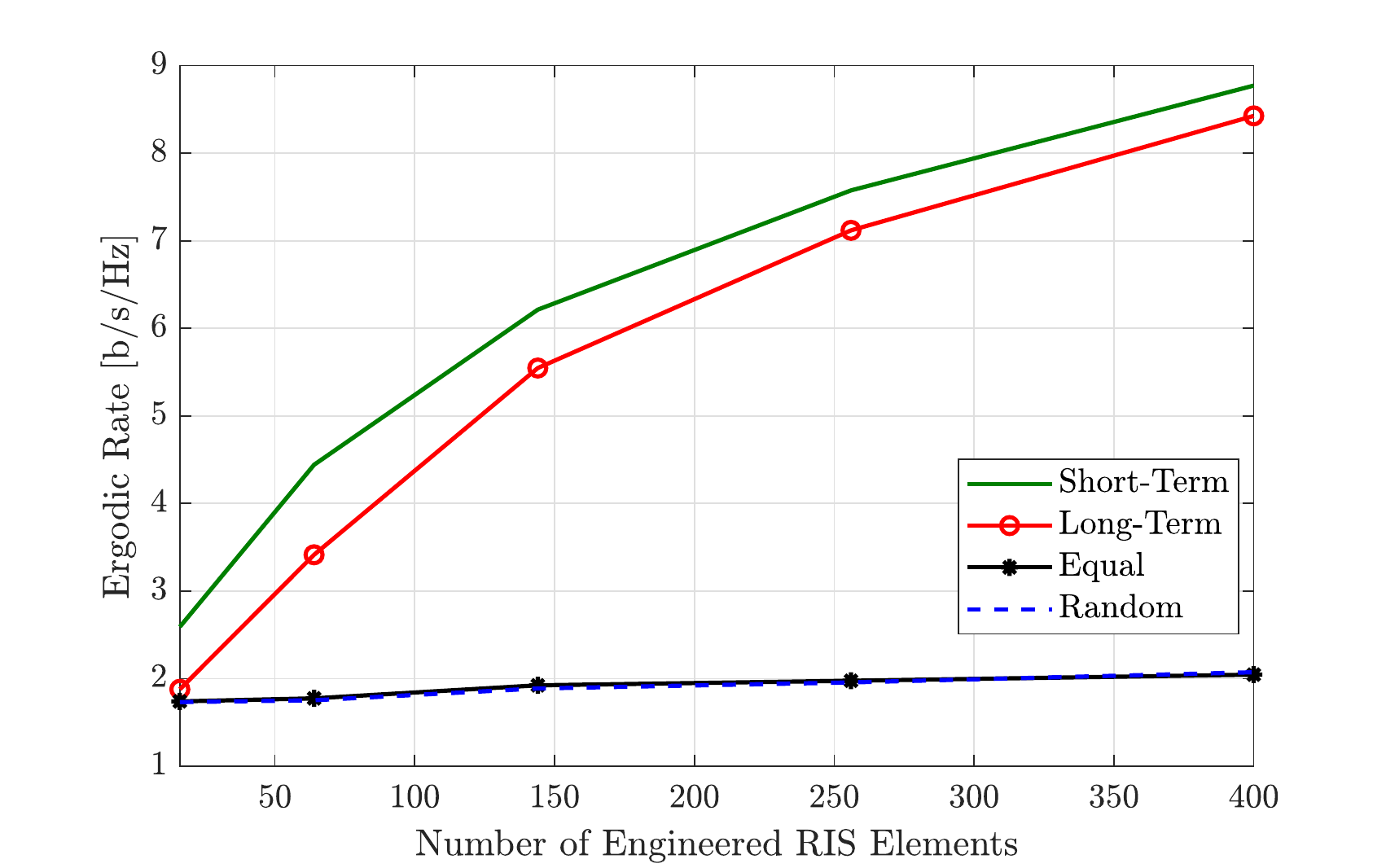} \vspace*{-0.4cm}\\
		(c)
		\vspace*{-0.3cm}
	\end{minipage}
	\caption{The system performance: $(a)$ Coverage probability v.s. the target rate [b/s/Hz] with $M=100$ and the destination is located at $(180, 15, 15)$~m; $(b)$ Average coverage probability (defined as $\mathbb{E} \{ \mathrm{Pr}(R > \xi) \}$, which is averaged over different realizations of the destination's location) v.s. the number of  RIS elements; and $(c)$ Ergodic rate v.s. the number of RIS elements.} \label{Fig1}
	\vspace*{-0.3cm}
\end{figure*}
\vspace*{-0.2cm}
\subsection{Coverage Probability} \label{Sub:CovP}
\vspace*{-0.2cm}
From the channel rate in \eqref{eq:Rran}, we define the coverage probability for a given target rate $\xi$ [b/s/Hz] as 
$	P_{\mathsf{cov}} = 1 - \mathrm{Pr}(R < \xi),$
where $\mathrm{Pr}(\cdot)$ denotes the probability of an event. By denoting $z = 2^\xi -1$, the coverage probability can be rewritten as
\vspace*{-0.2cm}
\begin{equation} \label{eq:PcovRan}
	P_{\mathsf{cov}} (z) =  1 - \mathsf{Pr}(\gamma < z),
\vspace*{-0.2cm}
\end{equation}
We utilize the moment-matching method to compute the coverage probability in \eqref{eq:PcovRan}.
\begin{theorem}~\label{Theorem:CovProbRan}
The coverage probability in \eqref{eq:PcovRan} can be formulated, in a closed-form expression, as
\vspace*{-0.2cm}
\begin{equation} \label{eq:Pcovlt}
	P_{\mathsf{cov}} (z) \approx \Gamma\left( k, z/w \right)/\Gamma (k),
\vspace*{-0.2cm}
\end{equation}
where the shape parameter $k$ and the scale parameter $w$ depend on the criterion for optimizing the phase shifts of the RIS elements. If the long-term phase shift design is utilized, we have
\vspace*{-0.2cm}
\begin{align}
	k
	&= \frac{ \left(\beta_{\mathrm{sd}} + 
		 o_{1} \beta_{\mathrm{sr}} \beta_{\mathrm{rd}} 
		\right)^2 }{\beta_{\mathrm{sd}}^2 +
		o_{2}  \beta_{\mathrm{sr}}^2 \beta_{\mathrm{rd}}^2  + 2 \beta_{\mathrm{sd}} o_{1} \beta_{\mathrm{sr}} \beta_{\mathrm{rd}} 
	},
	\\
	w
	&=
	\frac{
		\nu \beta_{\mathrm{sd}}^2 +
		\nu^2 o_{2}  \beta_{\mathrm{sr}}^2 \beta_{\mathrm{rd}}^2 + 2 \nu \beta_{\mathrm{sd}} o_{1} \beta_{\mathrm{sr}} \beta_{\mathrm{rd}}
	}{
		 \beta_{\mathrm{sd}} + 
		 o_{1} \beta_{\mathrm{sr}} \beta_{\mathrm{rd}} 
	},
\vspace*{-0.2cm}
\end{align}
where $\nu = \rho /\sigma^2$ is the ratio between the transmit power and noise power, and the scalars $o_1$ and $o_2$ are defined as
\vspace*{-0.2cm}
\begin{align}
 o_{1} &= \omega^{-1}({{K_{\mathrm{sr}}}{K_{\mathrm{rd}}} \eta + \widetilde{K} M}), \label{eq:o1}\\
 o_2 &= \omega^{-2}(2\eta K_{\mathrm{sr}}K_{\mathrm{rd}}( {M\widetilde{K} + 4} ) +  M^2\widetilde{K}^2 + 2M( {2\widetilde K - 1} )), \label{eq:o2}
 \vspace*{-0.2cm}
\end{align}
with $\eta = | \bar{\mathbf{h}}_{\mathrm{sr}}^H  \mathbf{\Phi} {\bar{\mathbf{h}} }_{\mathrm{rd}} |^2$. 

If, on the other hand, the short-term phase shift design is considered, the shape and scale parameters are
\vspace*{-0.2cm}
\begin{align}
	k &= \frac{k_c ( k_c + 1 ) }{2 ( 2k_c + 3 )},	w  = 2 \nu w_c^2\left( 2k_c + 3 \right),
\vspace*{-0.2cm}
\end{align}
where $k_c = \frac{( {c_1} + {c_2}\sqrt {{\beta_{\mathrm{sr}}}{\beta_{\mathrm{rd}}}} )^2}{c_3 + c_4\beta_{\mathrm{sr}}\beta_{\mathrm{rd}}}$ and  $w_c = \frac{c_3 + c_4\beta_{\mathrm{sr}}\beta_{\mathrm{rd}}}{{c_1} + {c_2}\sqrt {{\beta_{\mathrm{sr}}}{\beta_{\mathrm{rd}}}}} $ with
\vspace*{-0.2cm}
\begin{align}
& c_1 = {0.5\sqrt {\pi\beta_{\mathrm{sd}}} }, c_2 = 0.25 M \pi t_{\mathrm{sr}} t_{\mathrm{rd}} \omega^{-0.5} ,   \\
& 
c_3 = \frac{4 - \pi }{4}\beta_{\mathrm{sd}}, c_4 = M
- \frac{M{{\pi ^2}}}{{16}} t_{\mathrm{sr}}^2 t_{\mathrm{rd}}^2 \omega^{-1},
\vspace*{-0.2cm}
\end{align}
and $t_{\mathrm{sr}} =  {_1{F_1}\left( { - 0.5,1, - {K_{\mathrm{sr}}}} \right)}$, $t_{\mathrm{rd}} = {_1{F_1}\left( { - 0.5,1, - {K_{\mathrm{rd}}}} \right)}$ with $_1F_1(\cdot, \cdot, \cdot)$ being the confluent hypergeometric function of the first kind.
\end{theorem}
\begin{proof}
The proof is based on computing the mean and the variance of the SNR expressions in \eqref{eq:SNRLT} and \eqref{eq:SNRST}. To this end, the closed-form expressions in Theorem~\ref{Theorem:ChannelStatistics} is used. Then, the obtained mean and variance are matched to those of a Gamma distribution. The details of the proof are omitted due to space limitations.
\vspace*{-0.1cm}
\end{proof}
The coverage probability in \eqref{eq:Pcovlt} offers a simple closed-form expression for evaluating the performance of RIS-assistusually ed communications without the need of resorting to Monte Carlo simulations.
 As the number of RIS elements is usually sufficiently large, the obtained analytical expressions can be further simplified. 
If, for example, we ignore the nondominant terms, the shape parameter tends to $k \rightarrow ({{K_{\mathrm{sr}}}{K_{\mathrm{rd}}} \eta + \widetilde{K} M})^2 / ( 2\eta K_{\mathrm{sr}}K_{\mathrm{rd}} {M\widetilde{K}}  +  M^2\widetilde{K}^2)$ for the long-term phase shift design and to $k \rightarrow \left( {c_1} + {c_2}\sqrt {{\beta_{\mathrm{sr}}}{\beta_{\mathrm{rd}}}}  \right)^2 /(4c_3 + 4c_4{\beta_{\mathrm{sr}}}{\beta_{\mathrm{rd}}})$ for the short-term phase shift design.
If $M \rightarrow \infty$, in addition, we obtain $k \rightarrow 1$ and $k \rightarrow \infty $ for the long-term and short-term phase shift designs, respectively.
Furthermore, the scale parameter tends to $w \rightarrow  \nu^2 o_2 \beta_{\mathrm{sr}}  \beta_{\mathrm{rd}} / o_1$ and to $w \rightarrow 4\nu (c_3 + c_4 \beta_{\mathrm{sr}} \beta_{\mathrm{rd}} )$ for the long-term and short-term phase shift designs, respectively.
If $M \rightarrow \infty$, we obtain $w \rightarrow \infty$. By rewriting the Gamma function in \eqref{eq:Pcovlt} in a series expression and ignoring the high-order terms, the coverage probability in \eqref{eq:Pcovlt} is simplified to 
$P_{\mathsf{cov}} \rightarrow 1 - \frac{(z/w)^k}{k^2 \Gamma(k)}$ \cite{9195523}.
Based on the obtained values of $k$ and $w$, this implies that the coverage probability tends to $1$ as $M \rightarrow \infty$ for both the short-term and the long-term phase shift designs. This implies that, if the number of reconfigurable RIS elements is large enough, an RIS is capable of offering a good coverage.
It is worth mentioning that the coverage probability in \eqref{eq:Pcovlt} can applied to arbitrary phase shift designs, including the random and equal phase shifts as reported in \cite{van2021outage}. 

\vspace*{-0.2cm}
\subsection{Ergodic Channel Rate}
\vspace*{-0.2cm}
The channel rate in \eqref{eq:Rran} depends on the small-scale and large-scale fading coefficients. In this section, we study the ergodic channel rate over a long time period by averaging out the small-scale fading as follows
\vspace*{-0.2cm}
\begin{equation} \label{eq:ErgodicRate1}
\bar{R} = \mathbb{E}\{ \log_2 ( 1 + \gamma)\} , \mbox{ [b/s/Hz]}.
\vspace*{-0.2cm}
\end{equation}
A closed-form expression of \eqref{eq:ErgodicRate1} is given in Lemma~\ref{lemma:ErgodicRate} that still relies on a moment-marching approach, according to which the received SNR is matched to a Gamma distribution.
\vspace*{-0.4cm}
\begin{lemma} \label{lemma:ErgodicRate}
The ergodic channel capacity in \eqref{eq:ErgodicRate1} can be formulated in the closed-form expression as follows: 
\vspace*{-0.2cm}
	\begin{align} \label{eq:ErgodicRate}
		\bar{R} = \frac{1}{{\Gamma \left( {{k}} \right)\ln \left( 2 \right)}}G_{2,3}^{3,1}\left( {\left. {\frac{1}{{{w}}}} \right|\begin{array}{*{20}{c}}
				{0,1}\\
				{0,0,{k}}
		\end{array}} \right),
	\vspace*{-0.2cm}
	\end{align}
where $G_{p,q}^{m,n}\Big( { z \Big|\begin{array}{*{20}{c}}
		{a_1,\ldots, a_q}\\
		{b_1,\ldots,b_p}
\end{array}} \Big)$ is  the  Meijer-G  function, and, similar to Theorem~\ref{Theorem:CovProbRan}, $k$ and $w$ are the shape parameter and the scale parameter of the approximating Gamma distribution, respectively.
\end{lemma}
\begin{proof}
The proof follows along the same lines as the proof in \cite{9195523}, with the exception that the channel model considered in this paper is different.
\end{proof}
Differently from \cite{9195523}, Lemma~\ref{lemma:ErgodicRate} can be applied to all phase shift designs, which include the short-term and the long-term phase shifts designs of interest in this paper. 
The analytical expressions for $k$ and $w$ that correspond to the latter two phase designs are the same as for Theorem~\ref{Theorem:CovProbRan}.

\vspace*{-0.2cm}
\section{Numerical Results}
\vspace*{-0.2cm}
In this section, we validate the obtained analytical frameworks with the aid of Monte Carlo simulations.
We consider an RIS-assisted wireless network 
where the source is located at the origin and the center-point of the RIS is located at $(27, 25, 25)$~m. 
For the direct link, the channel gain $\beta_{\mathrm{sd}}$ [dB] is $\beta_{\mathrm{sd}} = -33.1 - 3.50 \log_{10}(d_{\mathrm{sd}}/1 \mbox{m})$, where $d_{\mathrm{sd}}$ is the distance between the source and the destination. For the indirect link, the channel gains $\beta_\alpha$ [dB] are $\beta_\alpha = -25.5 - 2.4 \log_{10}(d_\alpha / 1 \mbox{m})$, where $d_\alpha$ is the distance between the transmitter and the receiver. 
The Rician factors are equal to $K_{\alpha} = 10^{1.3 - 0.003 d_\alpha}$. The transmit power is $20$~mW and the system bandwidth is $10$~MHz. The carrier frequency is $1.8$~GHz and the noise power is $-94$~dBm. The following phase shift designs are considered for comparison: $i)$ the short-term phase shift design in \eqref{eq:Phase1}; $ii)$ the long-term phase shift design in \eqref{eq:PhaseLoS}; $iii)$ the equal phase shift design where the phase shifts are all set equal to $\pi/4$; and $iv)$ the random phase shift design where arbitrary values in the range $[-\pi, \pi]$ are considered.

In Fig.~\ref{Fig1}(a), we compared the closed-form expression of the coverage probability against Monte Carlo simulations. The good overlap between the analytical results and the numerical simulations confirms the accuracy of \eqref{eq:Pcovlt}. 
In Fig.~\ref{Fig1}(b), we utilize the analytical framework in \eqref{eq:Pcovlt} to evaluate the coverage probability as a function of the number of RIS elements and for different designs of the phase shifts. 
The phase shift designs based on short-term and long-term CSI offer significant gains compared to the equal and random phase shift designs. In addition, the gap between the short-term and long-term phase shifts design reduces as the number of RIS elements increases.
Finally, Fig.~\ref{Fig1}(c) displays the ergodic rate [b/s/Hz] in \eqref{eq:ErgodicRate}. 
We evince that the deployment of an RIS results in a substantial increase of the ergodic rate, as opposed to surfaces that operate as random scatterers and are not smart and reconfigurable.
Notably, the long-term phase shift design provides an ergodic rate that is close to the short-term phase shift design, and approaches it if the number of RIS elements is sufficiently large.

\vspace*{-0.2cm}
\section{Conclusion}
\vspace*{-0.2cm}
This paper has investigated the coverage probability and the ergodic rate of an RIS-assisted link for different phase shift designs depending on the amount of channel information that is exploited for optimizing the RIS. 
It is shown that a long-term phase shift design that depend on long-term CSI offers a suboptimal solution with a performance loss compared with the optimal phase shift design based on perfect CSI that decreases as the number of RIS elements increases. 
Generalization of this research work includes the analysis of sources and destinations equipped with multiple antennas. 

\vfill\pagebreak
\bibliographystyle{IEEEtran}
\bibliography{IEEEabrv,refs}

\end{document}